%% file: main.tex
\documentclass[runningheads]{llncs}

\input{header}

\input{macros}

\makeatletter
\def\thanks#1{\protected@xdef\@thanks{\@thanks
        \protect\footnotetext{#1}}}
\makeatother

\begin{document}

\title{Nonlinear Craig Interpolant Generation over Unbounded Domains by Separating Semialgebraic Sets}
\thanks{The first two authors marked with $\star$ contributed equally to this work and should be considered co-first authors. This work has been partially funded   by the National Key R\&D Program of China under grant No.\ 2022YFA1005101 and 2022YFA1005102, by the NSFC under grant No.\ 62192732, 62032024, and 12201618, by the CAS Project for Young Scientists in Basic Research under grant No.\ YSBR-040, by the Key R\&D Program of Hubei Province (2023BAB170), and by the Fundamental Research Funds for the Central Universities 
}


\author{
Hao Wu$^\star$
\inst{1} \orcidlink{0000-0001-9368-4744} \and
Jie Wang$^\star$\inst{2} \orcidlink{0000-0002-9681-1451} \and
Bican Xia \inst{3} \orcidlink{0000-0002-2570-2338} \and
Xiakun Li \inst{3} \orcidlink{0009-0007-1663-1287} \and
Naijun Zhan \inst{4,1} \orcidlink{0000-0003-3298-3817} \and
Ting Gan \inst{5} \inst{(}\Envelope\inst{)} \orcidlink{0000-0002-4880-5129}
}

\institute{
Institute of Software, Chinese Academy of Sciences, University of CAS, China \and 
Academy of Mathematics and Systems Science, Chinese Academy of Sciences\and 
School of Mathematical Sciences, Peking University, Beijing, China \and
School of Computer Science, Peking University, Beijing, China \and
School of Computer Science, Wuhan University, China
}

\maketitle

\input{sections/0-abstract}
\input{sections/1-introduction}
\input{sections/2-preliminaries}
\input{sections/3-existence}
\input{sections/4-sos}
\input{sections/5-interpolation}
\input{sections/6-conclusion}


\bibliographystyle{splncs04}
\bibliography{refer}

\input{sections/appendix}
\end{document}

%% file: header.tex
\usepackage{amsfonts}

\usepackage{latexsym,amsmath,amsthm,epsfig,amssymb,graphics,amscd,amsfonts}
\usepackage{float}
\usepackage{comment}
\usepackage{verbatim}
\usepackage{color}
\usepackage[usenames,dvipsnames,table]{xcolor}
\usepackage{indentfirst}
\usepackage{latexsym,amsmath,epsfig,amssymb,graphics}
\allowdisplaybreaks
\usepackage{stmaryrd}
\usepackage{mathrsfs}
\usepackage{algorithm, algorithmic}
\usepackage{amsfonts}
\usepackage{makeidx}
\usepackage{mathtools}
\usepackage{makecell}
\usepackage{multirow}
\usepackage{siunitx}
\usepackage{caption}
\usepackage{subcaption}
\usepackage{wrapfig}

\usepackage[hyperfootnotes=true,colorlinks,linkcolor=RedViolet,anchorcolor=RedViolet,citecolor=blue,urlcolor=black]{hyperref}
\usepackage[capitalize]{cleveref}
\crefname{lemma}{Lem.}{Lem.}
\creflabelformat{lemma}{#2\textup{#1}#3}
\crefname{example}{Exmp.}{Exmp.}
\creflabelformat{example}{#2\textup{#1}#3}
\crefname{section}{Sect.}{Sect.}
\creflabelformat{section}{#2\textup{#1}#3}
\crefname{appendix}{Appx.}{Appx.}
\creflabelformat{appendix}{#2\textup{#1}#3}
\crefname{definition}{Def.}{Def.}
\creflabelformat{definition}{#2\textup{#1}#3}
\crefname{theorem}{Thm.}{Thm.}
\creflabelformat{theorem}{#2\textup{#1}#3}
\crefname{proposition}{Prop.}{Prop.}
\creflabelformat{proposition}{#2\textup{#1}#3}
\crefname{corollary}{Cor.}{Cor.}
\creflabelformat{corollary}{#2\textup{#1}#3}
\crefname{problem}{Problem}{Problem}
\creflabelformat{problem}{#2\textup{#1}#3}
\crefname{algorithm}{Alg.}{Alg.}
\creflabelformat{algorithm}{#2\textup{#1}#3}
\crefname{equation}{Eq.}{Eq.}
\crefname{enumi}{}{}

\usepackage[misc]{ifsym}
\usepackage{orcidlink}
\usepackage{bbding}

\RequirePackage{todonotes,xspace}
\newcommand{\oomit}[1]{ }

%% file: macros.tex
\newcommand{\RR}{\mathbb{R}}
\newcommand{\NN}{\mathbb{N}}
\newcommand{\TT}{\mathcal{T}}

\newcommand{\xx}{\mathbf{x}}
\newcommand{\yy}{\mathbf{y}}
\newcommand{\zz}{\mathbf{z}}
\def\ba{{\boldsymbol{\alpha}}}
\def\cl{\hbox{\rm{cl}}}

\newtheorem{proper}{Property}
\newtheorem*{problem*}{Problem}


%% file: sections/0-abstract.tex
\begin{abstract}
Interpolation-based techniques become popular in recent years, as they can improve the scalability of existing verification techniques due to their inherent modularity and local reasoning capabilities. 
Synthesizing Craig interpolants is the cornerstone of these techniques. 
In this paper, we investigate nonlinear Craig interpolant synthesis for two polynomial formulas of the general form, essentially corresponding to the underlying mathematical problem to separate two disjoint semialgebraic sets.
By combining the homogenization approach with existing techniques, we prove the existence of a novel class of non-polynomial interpolants called semialgebraic interpolants. 
These semialgebraic interpolants subsume polynomial interpolants as a special case. 
To the best of our knowledge, this is the first existence result of this kind.
Furthermore, we provide complete sum-of-squares characterizations for both polynomial and semialgebraic interpolants, which can be efficiently solved as semidefinite programs.
Examples are provided to demonstrate the effectiveness and efficiency of our approach. 

\keywords{Craig interpolation \and Separating semialgebraic sets \and Homogenization \and Sum-of-squares \and Semidefinite programming}
\end{abstract}

%% file: sections/1-introduction.tex
\section{Introduction} \label{sec:intr}

\paragraph{Background.}
Craig interpolant is a fundamental concept in formal verification and automated theorem proving. 
It was introduced by William Craig in the 1950s as a tool for reasoning about logical formulas and their satisfiability. 
Craig interpolation techniques possess excellent modularity and local reasoning capabilities, 
making them effective tools for enhancing the scalability of formal verification methods, like theorem proving \cite{krajicek97,pudlak97}, model-checking \cite{mcmillan03}, abstract interpretation \cite{HJMM04,mcmillan05}, program verification \cite{wang11,LSXLSH2017} and so on. 


Efficient generation of Craig interpolants is crucial in interpolation-based techniques, 
and therefore has garnered increasing attention.
Formally, a formula $I$ is called a Craig interpolant for two mutually exclusive formulae $\phi$ and $\psi$ in a background theory $\TT$, if it is defined on the common symbols of $\phi$ and $\psi$, implied by $\phi$ in the theory $\TT$, and inconsistent with $\psi$ in the theory $\TT$.
Due to the diversity of background theories and their integration, researchers have been dedicated to developing efficient interpolation synthesis algorithms.
Currently, numerous effective algorithms for automatic synthesis of interpolants have been proposed for various fragments of first-order logic, e.g., linear arithmetic \cite{HJMM04}, logic with arrays \cite{hoenicke18ijcar,mcmillan08}, logic with sets \cite{KMZ06}, equality logic with uninterpreted functions (EUF) \cite{mcmillan05,CGS08}, etc., and their combinations \cite{YM05,KV09,RS10}.
Moreover, D'Silva et al. \cite{SPWK10} explored how to compare the strength of various interpolants.   

However, interpolant generation for nonlinear arithmetic and its combination with the aforementioned theories is still in infancy, although nonlinear polynomial inequalities are quite common in software involving number theoretic functions as well as hybrid systems \cite{Zhan17,ZZKL12}. 
In addition, when the formulas $\phi$ and $\psi$ are defined by polynomial inequalities, generating an interpolant is essentially equivalent to the mathematical problem of separating two disjoint semialgebraic sets, which has a long history and is a challenging mathematical problem \cite{acquistapace1999separation}.


In~\cite{DXZ13}, Dai et al. attempted to generate interpolants for conjunctions of mutually contradictory nonlinear polynomial inequalities without unshared variables. 
They proposed an algorithm based on Stengle's Positivstellensatz \cite{Stengle}, which guarantees the existence of a witness and can be computed using semidefinite programming (SDP). 
While their algorithm is generally incomplete, it becomes complete when all variables are bounded, known as the Archimedean condition (see in \cref{ssec:QM}). 

In \cite{GDX16}, Gan et al. introduced an algorithm for generating interpolants specifically for quadratic polynomial inequalities. 
Their approach is based on the insight that analyzing the solution space of concave quadratic polynomial inequalities can be achieved by linearizing them, using a generalization of Motzkin's transposition theorem. 
Additionally, they discussed generating interpolants for a combination of the theory of quadratic concave polynomial inequalities and EUF using a hierarchical calculus proposed in \cite{SSLMCS2008} and employed in \cite{RS10}. 

In \cite{gan2020nonlinear}, Gan et al. further extended the problem from the case of quadratic concave inequalities to the more general Archimedean case. 
To accomplish this, they utilized Putinar's Positivstellensatz and proposed a Craig interpolation generation method based on SDP. 
This method allows to generate interpolants in a broader class of situations involving nonlinear polynomial inequalities. 
However, the Archimedean condition still imposes a limitation on the method, as it requires bounded domains.

In \cite{chen19cade}, Chen et al. proposed a counterexample-guided framework based on support vector machines for synthesizing nonlinear interpolants.
Later in \cite{lin2022}, Lin et al. combined this framework and deep learning for synthesizing neural interpolants.
In \cite{jovanovic21cav}, Jovanović and Dutertre also designed a counterexample-guided framework based on cylindrical algebraic decomposition (CAD) for synthesizing interpolants as boolean combinations of constraints.
However, these approaches rely on quantifier elimination to ensure completeness and convergence, which terribly affects their efficiency due to its doubly exponential time complexity~\cite{dh88}.

For theories including non-polynomial expressions, the general idea is to abstract non-polynomial expressions into polynomial or linear expressions.
In \cite{GZ16}, Gao and Zufferey presented an approach for extracting interpolants for nonlinear formulas that may contain transcendental functions and differential equations. 
They accomplished this by transforming proof traces from a $\delta$-decision procedure~\cite{Gao14} based on interval constraint propagation (ICP)~\cite{BG06}. 
Like the Archimedean condition, $\delta$-decidability also imposes the restriction that variables are bounded (in a hyper-rectangle).
A similar idea was also reported in \cite{KB11}.
In \cite{srikanth2017} and \cite{cimatti18tcl}, Srikanth et al. and Cimatti et al. proposed approaches to abstract nonlinear formulas into the theory of linear arithmetic with uninterpreted functions.

\paragraph{Contributions.} 
In this paper, we consider how to synthesize an interpolant function $h(\xx)$ for two polynomial formulas $\phi(\xx,\yy)$ and $\psi(\xx,\zz)$ such that $\phi(\xx,\yy)\models h(\xx)>0$ and $\psi(\xx,\zz)\models h(\xx)<0$ without assuming the Archimedean condition, i.e., the variables in $\phi$ and $\psi$ can have an unbounded range of values.
Here, uncommon variables of $\phi$ and $\psi$ are allowed, and the description of formulas may involve any polynomial of any degree. Hence the problem is more general than the ones discussed in \cite{DXZ13,GDX16,gan2020nonlinear}, and is also more difficult as polynomial interpolants may not exist \cite{acquistapace1999separation}.
To address this problem, we first utilize homogenization techniques to elevate the descriptions of $\phi$ and $\psi$ to the homogeneous space. 
In this homogeneous space, we can impose the constraint that the variables lie on a unit sphere, thus reviving the Archimedean condition. 
Combining this idea with the work in \cite{gan2020nonlinear}, we can prove the existence of a semialgebraic function $h(\xx)=h_1(\xx)+h_2(\xx)\sqrt{\|\xx\|^2+1}$ such that $h(\xx)>0$ serves as an interpolant, where $h_1,h_2$ are polynomials ($h$ becomes a polynomial when $h_2=0$).
Furthermore, we provide sum-of-squares (SOS) programming procedures for finding such semialgebraic interpolants as well as polynomial interpolants.
Under certain assumptions, we prove that the SOS procedures are sound and complete.

\paragraph{Organization.} 
The rest of the paper is organized as follows. 
Preliminaries are introduced in \cref{sec:prel}.
\cref{sec:existence} proves the existence of an interpolant for two mutually contradictory polynomial formulas.
\cref{sec:sos} derives an SOS characterization for the interpolant.
\cref{sec:5} presents an SDP-based method for computation and provides examples with portraits.
Finally, \cref{sec:con}, we conclude this paper and discuss some future works.
Omitted proofs are given in the \cref{app}.

%% file: sections/2-preliminaries.tex
\section{Preliminaries} \label{sec:prel}

We first fix some basic notations. 
Let $\mathbb{N}$ and $\mathbb{R}$ be the sets of integers and real numbers, respectively. 
By convention, we use boldface letters to denote vectors.
Fixing a vector of indeterminates $\xx\coloneqq(x_1,\ldots,x_r)$, let $\mathbb{R}[\xx]$ denote the polynomial ring in variables $\xx$ over real numbers. 
We use $\mathrm{\Sigma}[\xx]\coloneqq\{\sum_{i=1}^m p_i^2\mid p_i\in \mathbb{R}[\xx], m\in \NN\}$ to denote the set of SOS polynomials in variables $\xx$.  
A basic semialgebraic set $S\subseteq \mathbb R^r$ is of the form $\{\xx\in \mathbb{R}^r \mid p_1(\xx) \triangleright 0, \dots, p_m(\xx) \triangleright 0\}$, where $p_i(\xx) \in \mathbb{R}[\xx]$ and $\triangleright\in \{\ge, >\}$
(each of the inequalities can be either  
strict or non-strict).
A basic semialgebraic set is said to be \emph{closed} if it is defined by non-strict inequalities.
Semialgebraic sets are formed as unions of basic semialgebraic sets. i.e., $T=\bigcup_{i=1}^{n} S_i$ is a semialgebraic set, where each $S_i$ is a basic semialgebraic set.
For any (semialgebraic) set $S\subseteq \mathbb{R}^r$, let $\cl(S)$ denote the closure of $S$. 
Let $\bot$ and $\top$ stand for \textbf{false} and \textbf{true}, respectively.
For a vector $\xx\in\mathbb{R}^r$, let $\|\xx\|\coloneqq \sqrt{\sum_{i=1}^r x_i^2}$ denote the standard Euclidean norm. 

In the following, we give a brief introduction on important notions used throughout the rest of this paper and then describe the problem of interest.

\subsection{Quadratic Module} 
\label{ssec:QM}
\label{sn}

\begin{definition}[Quadratic Module \cite{marshall2008book}]
A subset $\mathcal{M}$ of $\mathbb{R}[\xx]$ is called a \emph{quadratic module} if it contains 1 and  is closed under addition and  multiplication with squares, i.e.,
\[ 1\in\mathcal{M}, \mathcal{M} + \mathcal{M} \subseteq \mathcal{M}, \mbox{ and } p^2 \mathcal{M} \subseteq \mathcal{M}~ \mbox{ for all } p \in \mathbb{R}[\xx].\]
\end{definition}
\begin{definition}
  Let $\overline{p}:=\{p_1,\ldots,p_m\}$ be a finite subset of $\mathbb{R}[\xx]$. The quadratic module $\mathcal{M}_{\xx}( \overline{p})$, or simply $\mathcal{M}(\overline{p})$, generated by $\overline{p}$ is the smallest quadratic module containing
  all $p_i$, i.e.,
  \begin{equation*}
  \mathcal{M}_{\xx}(\overline{p})\coloneqq \{\sigma_0+\sum_{i=1}^m \sigma_i p_i \mid \sigma_0, \sigma_i \in \mathrm{\Sigma}[\xx]\}.    
  \end{equation*}
\end{definition}
Let $S$ be a closed basic semialgebraic set described by $\overline{p}\ge \boldsymbol{0}$, i.e., 
\begin{equation}\label{eq:S}
    S \coloneqq \{\xx \in\RR^r \mid p_1(\xx) \ge 0, \ldots, p_m (\xx) \ge 0\}.
\end{equation}
Since SOS polynomials are non-negative, it is easy to verify that the quadratic module $\mathcal{M}(\overline{p})$ is a subset of polynomials that are nonnegative on $S$. 
In fact, under the so-called the Archimedean condition, the quadratic module $\mathcal{M}(\overline{p})$ contains all polynomials that are strictly positive over $S$.
Both the condition and the statement are formalized as follows.

\begin{definition}[Archimedean \cite{marshall2008book}]
  Let $\mathcal{M}$ be a quadratic module of $\mathbb{R}[\xx]$. $\mathcal{M}$ is said to be \emph{Archimedean} if there exists some $a>0$ such that $a-\|\xx\|^2 \in \mathcal{M}$. 
  Furthermore, if $\mathcal{M}(\overline{p})$ is Archimedean, we say that the semialgebraic set $S$ as defined in \cref{eq:S} is of the \emph{Archimedean form}.
\end{definition}

\begin{theorem}[Putinar's Positivstellensatz \cite{putinar93}]
\label{thm:putinar}
    Let $\overline{p}\coloneqq \{p_1,\ldots,p_m\}$ and 
    $S$ be defined in \cref{eq:S}.
    Assume that the quadratic module $\mathcal{M}(\overline{p})$ is Archimedean. If $f(\xx)> 0$ over $S$, then $f\in \mathcal{M}(\overline{p})$.
\end{theorem}

The above theorem serves as a key result in real algebraic geometry, as it provides a simple characterization of polynomials that are locally positive on closed basic semialgebraic sets.
Because of this, \cref{thm:putinar} is widely used in the field of polynomial optimization, referring to \cite{lasserre09book,marshall2008book} for an in-depth treatment of this topic.

Though powerful, \cref{thm:putinar} relies on the Archimedean condition.
Note that the inclusion $a-\|\xx\|^2 \in \mathcal{M}$ implies that $a-\|\xx\|^2 \ge 0$ over $S$, deducing that $S$ is contained in a ball with radius~$\sqrt{a}$.
As a result, in case that the set $S$ is unbounded, \cref{thm:putinar} is not directly applicable.

\subsection{Homogenization} \label{ssec:homo}
Let $\xx=(x_1, \ldots,x_r)\in\RR^{r}$ be an $r$-tuple of variables and
$x_0$ a fresh variable.
Suppose that $f(\xx)\in\RR[\xx]$ is a polynomial of degree $d_f$.
We denote by $\tilde{f}(x_0,\xx)\in\RR[x_0,\xx]$ the homogenization of $f(\xx)$
which is obtained by substituting $\frac{x_1}{x_0}$ for $x_1$, ..., $\frac{x_r}{x_0}$ for $x_r$ in $f(\xx)$ and then multiplying with $x_0^{d_f}$, that is, 
\begin{equation} \label{eq:hom:var}
  \tilde{f}(x_0,\xx)\coloneqq 
  x_0^{d_f} f(\frac{x_1}{x_0}, \ldots, \frac{x_r}{x_0}).
\end{equation}
For example, if $f(\xx)=x_1^3+2x_1x_2+3x_2+4$, then $\tilde{f}(x_0,\xx)=x_1^3+2x_0x_1x_2+3x_0^2x_2+4x_0^3$.
In what follows, we always use the variable $x_0$ as the homogenizing variable.

Let $S$ be defined as in \cref{eq:S}. We define the following set related to $S$ by homogenizing polynomials in the description of $S$:
\begin{equation} \label{eq:tSH}
    \tilde{S}^h\coloneqq\{(x_0,\xx)\in\RR^{r+1}\mid \tilde{p}_1(x_0, \xx) \ge 0 ,\ldots,\tilde{p}_m (x_0, \xx) \ge 0,x_0>0, x_0^2+\|\xx\|^2=1\}.
\end{equation}
Obviously, the following property holds.
\begin{proper} \label{property:hom}
  Let $S$ be as in \cref{eq:S} and $\tilde{S}^h$ be defined as above. Then, $\xx \in S$ if and only if
  \begin{equation*}
    \left(\frac{1}{\sqrt{1+\|\xx\|^2}}, \frac{x_1}{\sqrt{1+\|\xx\|^2}}, \ldots, \frac{x_r}{\sqrt{1+\|\xx\|^2}}\right)\in \tilde{S}^h.
  \end{equation*}
  Moreover, $(x_0, \xx) \in \tilde{S}^h$ if and only if 
  $(\frac{x_1}{x_0}, \ldots, \frac{x_r}{x_0})\in S$. 
\end{proper}
\begin{proof}
It is straightforward to verify.
\end{proof}

\cref{property:hom} shows that there exists a one-to-one correspondence between points in $S\in \RR^n$ and those in $\tilde{S}^h\in \RR^{n+1}$. 

We also define the set $\tilde{S}$ by replacing $x_0>0$ in \cref{eq:tSH} with $x_0\ge 0$:
\begin{equation} \label{eq:tS}
  \tilde{S}\coloneqq\{ (x_0, \xx)\in\RR^{r+1}\mid \tilde{p}_1(x_0, \xx) \ge 0,\ldots,\tilde{p}_m (x_0, \xx) \ge 0,x_0\ge0, x_0^2 + \|\xx\|^2=1\}.
\end{equation}
To capture the relation between $\tilde{S}^h$ and $\tilde{S}$, we introduce the following definition and a related useful lemma.

\begin{definition}\label{def:closed}
A closed basic semialgebraic set $S$ is closed at $\infty$ if $\cl(\tilde{S}^h)=\tilde{S}$.
\end{definition}

\begin{lemma}[\cite{huang2023homogenization}]\label{sec2:lm} 
Let $f\in\RR[\xx]$ and $S$ be a closed basic semialgebraic set. Then $f\ge0$ on $S$ if and only if $\tilde{f}\ge0$ on $\cl(\tilde{S}^h)$. Moreover, assuming that $S$ is closed at $\infty$, then $f\ge0$ on $S$ if and only if $\tilde{f}\ge0$ on $\tilde{S}$.
\end{lemma}


Let us define
\begin{align} \label{eq:Sinf}
    S^{(\infty)}&\coloneqq\{\xx\in\RR^{r} \mid p_1^{(\infty)}(\xx) \ge 0,\ldots,p_m^{(\infty)}(\xx) \ge 0, \|\xx\|^2=1\},
\end{align}
where $p^{(\infty)}(\xx)$ denotes the highest degree homogeneous part of a polynomial $p(\xx)\in \RR[\xx]$, e.g.,
if $p=x_1^2+2x_1x_2+3x_2^2+4x_1+5x_2$, then $p^{(\infty)}=x_1^2+2x_1x_2+3x_2^2$.


\begin{proper}
  \label{prop:3}
  Let $\tilde{S}^h$, $\tilde{S}$ and $S^{(\infty)}$ be defined as above. If $S^{(\infty)}$ is empty, then $\tilde{S}^h=\tilde{S}$.
\end{proper}
\begin{proof}
  It is straightforward to verify.
\end{proof}

\subsection{Problem Description} \label{ssec:prob}

Given two formulas $\phi$ and $\psi$ in a first-order theory $\TT$ s.t. $\phi \models \psi$, Craig showed that there always exists an \emph{interpolant} $I$ over the common symbols of $\phi$ and $\psi$ s.t.  $\phi \models I$ and $I \models \psi$. 
In the context of verification, we slightly abuse the terminology following \cite{mcmillan05}: 
A \emph{reverse interpolant} (as coined in \cite{KV09}) $I$ over the common symbols of $\phi$ and $\psi$ is defined as follows.
\begin{definition}[Interpolant]
\label{interpolant}
  Given two formulas $\phi$ and $\psi$ in a theory $\TT$ s.t. $\phi \wedge \psi \models_{\TT} \bot$,
  a formula $I$ is an \emph{interpolant} of $\phi$ and $\psi$ if 
  (1) $\phi \models_{\TT} I$;
  (2) $I \wedge \psi \models_{\TT} \bot$; and 
  (3) $I$ only contains common symbols and free variables
  shared by $\phi$ and $\psi$.
\end{definition}

The interpolant synthesis problem of interest in this paper is formulated as follows.

\begin{problem} \label{problem:1}
Let $\phi(\xx,\yy)$ and $\psi(\xx,\zz)$ be two polynomial formulas of the form
\begin{align}
  \phi(\xx,\yy) &\coloneqq \bigvee_{k=1}^{K_\phi} \bigwedge_{i=1}^{m_k} f_{k,i}(\xx,\yy) \ge 0, \label{eq:phi}\\
  \psi(\xx,\zz) &\coloneqq \bigvee_{k'=1}^{K_\psi} \bigwedge_{j=1}^{n_{k'}} g_{k',j}(\xx,\zz) \ge 0, \label{eq:psi}
 \end{align}
where $\xx \in \mathbb{R}^{r_1}$, $\yy \in \mathbb{R}^{r_2}$, $\zz \in \mathbb{R}^{r_3}$ are variable vectors, $r_1,r_2,r_3 \in \mathbb{N}$,
and $f_{k,i}, g_{k',j}$ are polynomials in the corresponding variables.
We aim to find a function $h(\xx)$ such that $h(\xx)>0$ is an interpolant for $\phi$ and $\psi$, i.e.,
\begin{equation*}
    \phi(\xx,\yy) \models h(\xx)>0 \text{ and } \psi(\xx,\zz) \models h(\xx)<0.
\end{equation*}
Here $h(\xx)$ is called an interpolant function. 
Specifically, we are interested in two scenarios where
\begin{enumerate}
    \item \textbf{Polynomial interpolants}: the function $h(\xx)$ is a polynomial in $\RR[\xx]$;
    \item \textbf{Semialgebraic interpolants}\footnote{A function $f(\xx)$ is called semialgebraic if its graph $\{(\xx,f(\xx))\mid \xx\in \RR^r\}$ is a semialgebraic set. 
    The graph of $h(\xx)$ is $\{\xx\in \RR^r\mid \exists w.~h(\xx)=h_1(\xx)+w\cdot h_2(\xx)\wedge w^2=1+\|\xx\|^2\wedge w\ge 0\}$.
    }: the function $h(\xx)$ can be expressed as 
    \begin{equation}\label{eq:h}
    h(\xx) = h_1(\xx) + \sqrt{\|\xx\|^2+1}\cdot h_2(\xx),
    \end{equation}
    with $h_1(\xx), h_2(\xx)\in \RR[\xx]$.
\end{enumerate}
Obviously, the second case degenerates to the first case when $h_2(\xx)=0$.
\end{problem}

\begin{remark}
    Like in \cite{gan2020nonlinear,Gao14}, we require $\phi$ and $\psi$ to be defined by \emph{non-strict} polynomial inequalities, mainly for two reasons:
    (1)~Theoretically, our approach relies on \cref{thm:putinar}, which necessitates a closed underlying basic semialgebraic set.
    (2)~Numerically, we employ numerical solvers incapable of distinguishing $\ge$ from $>$.
    In the coming sections, we will see the significance of both closedness and closedness at $\infty$ for the existence of interpolants.
\end{remark}

%% file: sections/3-existence.tex
\section{Existence of Interpolant}  \label{sec:existence}

In this section, we prove the existence of a semialgebraic interpolant function $h(\xx)$ of the form \cref{eq:h}, under certain conditions on $\phi$ and $\psi$.
In \cref{ssec:3.1}, we begin by focusing on the scenario where both $\phi$ and $\psi$ exclusively involve the variable $\xx$.
Subsequently, in \cref{ssec:3.2}, we expand our scope to the case where unshared variables, $\yy$ and $\zz$, emerge.
\vspace{-3mm}
\subsection{Interpolant between $\phi(\xx)$ and $\psi(\xx)$}
\label{ssec:3.1}

In this part, we prove the existence of a semialgebraic interpolant function of the form  in \cref{eq:h} that separates the two closed semialgebraic sets in $\RR^r$ corresponding to $\phi(\xx)$ and $\psi(\xx)$.
The basic idea goes as follows: 
First, we consider the problem of finding a semialgebraic function $h(\xx)$ such that $h(\xx)=0$ separates two closed basic semialgebraic sets $S_1$ and $S_2$ in $\mathbb{R}^{r}$.
Using the homogenization technique, we prove that there exists a polynomial $g\in \RR[x_0, \xx]$ with $g(x_0, \xx)=0$ separating $\tilde{S_1}$ and $\tilde{S_2}$, and the existence of $h(\xx)$ is directly induced by that of $g$ (see \cref{sec3:prop2}).
After that, we extend the result to the case where $S_1$ becomes a closed semialgebraic set (see \cref{sec3:basic-semi}) and when both $S_1$ and $S_2$ are closed semialgebraic sets (see \cref{sec3:thm1}).

We begin by recapping an existing result from \cite{gan2020nonlinear}.

\begin{proposition}[{\cite[Lem.~2]{gan2020nonlinear}}]
\label{sec3:prop1}
Let 
$S_1=\{\xx\in\RR^{r} \mid p_1(\xx) \ge 0,\ldots,p_m (\xx) \ge 0\}$, 
$S_2=\{\xx\in\RR^{r} \mid q_1(\xx) \ge 0,\ldots,q_n (\xx) \ge 0\}$ 
be two closed basic semialgebraic sets of the Archimedean form. 
Assuming that $S_1\cap S_2=\emptyset$, then there exists a polynomial $h(\xx)\in\RR[\xx]$ such that
\begin{equation}\label{sec3:eq1}
    \forall \xx\in S_1.~h(\xx)>0 \text{ and } \forall \xx\in S_2.~-h(\xx)>0.
\end{equation}
\end{proposition}


It is important to emphasize that the proof of \cref{sec3:prop1} relies on \cref{thm:putinar} and hence is limited to the case where the sets $S_1$ and $S_2$ are of the Archimedean form.
In the following \cref{sec3:prop2}, we show how to remove this restriction. 

\begin{proposition}\label{sec3:prop2}
Let $S_1=\{\xx\in\RR^{r} \mid p_1(\xx) \ge 0,\ldots,p_m (\xx) \ge 0\}$, $S_2=\{\xx\in\RR^{r} \mid q_1(\xx) \ge 0,\ldots,q_n (\xx) \ge 0\}$ be closed basic semialgebraic sets. Assuming that $\tilde{S}_1\cap \tilde{S}_2=\emptyset$, then there exists a semialgebraic function $h(\xx)$ of the form in \cref{eq:h} such that
\begin{equation}\label{sec3:eq2}
    \forall \xx\in S_1.~h(\xx)>0 \text{ and }\forall \xx\in S_2.~ -h(\xx)>0.
\end{equation}
\end{proposition}

\begin{proof}
By the definition of $\tilde{S}$ in \cref{eq:tS}, we know that $\tilde{S}_1$ and $\tilde{S}_2$ are two basic semialgebraic sets of the Archimedean form (as $1-x_0^2 - \|\xx\|^2$ belongs to the corresponding quadratic modules). 
Since $\tilde{S}_1\cap\tilde{S}_2=\emptyset$, by invoking \cref{sec3:prop1}, there exists a polynomial $g\in\RR[x_0, \xx]$ such that
\begin{equation} \label{eq:intgy}
    \forall (x_0, \xx)\in \tilde{S}_1.~g(x_0, \xx)>0 \text{ and }\forall (x_0,\xx)\in \tilde{S}_2.~-g(x_0, \xx)>0.
\end{equation}
Note that for any $\xx\in S_1$ (resp. $S_2$), by \cref{property:hom} we have $(\frac{1}{\sqrt{\|\xx\|^2+1}},\frac{\xx}{\sqrt{\|\xx\|^2+1}})\in\tilde{S}_1$ (resp. $\tilde{S}_2$).
Let 
\begin{equation}\label{eq:h-from-g}  h(\xx)\coloneqq(\sqrt{\|\xx\|^2+1})^{\deg(g)} g(\frac{1}{\sqrt{\|\xx\|^2+1}},\frac{\xx}{\sqrt{\|\xx\|^2+1}}).
\end{equation}
Since $(\sqrt{\|\xx\|^2+1})^{\deg(g)} \ge 1$, combining with \cref{eq:intgy}, we have that $h(\xx)$ satisfies \cref{sec3:eq2}.

To see that $h(\xx)$ admits the form in \cref{eq:h}, we expand the right-hand side of \cref{eq:h-from-g} and simplify the terms with power of $\sqrt{\|\xx\|^2+1}$ greater than or equal to $2$.
After simplification, we collect the terms with and without $\sqrt{\|\xx\|^2+1}$ into two groups so that $h(\xx)$ can be expressed as $h_1(\xx)+\sqrt{\|\xx\|^2+1}\cdot h_2(\xx)$ for $h_1(\xx),h_2(\xx)\in \RR[\xx]$.
\end{proof}

In order to check whether the condition $\tilde{S}_1\cap \tilde{S}_2=\emptyset$ in \cref{sec3:prop2} holds, one can use the following lemma.
\begin{lemma} \label{sec3:lm1}
Given two closed basic semialgebraic set $S_1$ and $S_2$, if $S_1\cap S_2=\emptyset$ and $S^{(\infty)}_1\cap S^{(\infty)}_2=\emptyset$, then $\tilde{S}_1\cap\tilde{S}_2=\emptyset$.
\end{lemma}

Now, we extend the result in \cref{sec3:prop2} to the case when $S_1$ and $S_2$ are two closed semialgebraic sets.
A closed semialgebraic set, say $T$, is a union of some closed basic semialgebraic sets, i.e., $T=\cup_{i=1}^a S_i$ with
\begin{equation*}
S_i=\{\xx\in\RR^r \mid p_{i1}(\xx) \ge 0,\ldots,p_{im_i} (\xx) \ge 0\},\quad i=1,\ldots,a,
\end{equation*}
where $p_{ik}(\xx) \in \RR[\xx]$, $m_i\in\NN$, $k=1,...,m_i$, $i=1,\ldots,a$.
Mirroring the definition of $S^{(\infty)}$ and $\tilde{S}$, we define $T^{(\infty)}\coloneqq\bigcup_{i=1}^a S_i^{(\infty)}$ and $\tilde{T}\coloneqq\bigcup_{i=1}^a\tilde{S}_i$.
In the following lemma, we deal with the case when $S_1$ in \cref{sec3:prop2} becomes a union of closed basic semialgebraic sets.

\begin{lemma}\label{sec3:basic-semi} 
Let $T_1=\cup_{i=1}^aS_{i}$ be a closed semialgebraic set with $S_i=\{\xx\in\RR^{r} \mid p_{i1}(\xx) \ge 0,\ldots,p_{im_i}(\xx) \ge 0\}$, and let $T_2=\{\xx\in\RR^{r} \mid q_1(\xx) \ge 0,\ldots,q_n (\xx) \ge 0\}$ be a closed basic semialgebraic set.
Assume that $\tilde{T}_1\cap \tilde{T_2}=\emptyset$. Then there exists a polynomial $g\in\RR[x_0, \xx]$ such that
\begin{equation}\label{sec3:eq3}
    \forall (x_0,\xx)\in \tilde{T}_1.~g(x_0,\xx)>0 \text{ and }\forall (x_0,\xx)\in \tilde{T}_2.~-g(x_0,\xx)>0.
\end{equation}
\end{lemma}

Then, we use \cref{sec3:basic-semi} to prove the case where both sets are unions of closed basic semialgebraic sets.

\begin{theorem}\label{sec3:thm1}
Let $T_1=\cup_{i=1}^a S_{i}$ and $T_2=\cup_{j=1}^b S'_{j}$ be closed semialgebraic sets, where $S_i$ and $S'_j$ are closed basic semialgebraic sets for $i=1,\dots, a, j=1,\dots,b$.
Assume $\tilde{T}_1\cap \tilde{T}_2=\emptyset$. Then there exists a semialgebraic function $h(\xx)$ of the form in \cref{eq:h} such that
\begin{equation}\label{sec3:eq4}
    \forall \xx\in T_1.~ h(\xx)>0 \text{ and }\forall \xx\in T_2.~ -h(\xx)>0.
\end{equation}
\end{theorem}

Similarly to \cref{sec3:lm1}, the condition $\tilde{T}_1\cap \tilde{T_2}=\emptyset$ can be verified by checking whether $T_1\cap T_2=\emptyset$ and $T^{(\infty)}_1\cap T^{(\infty)}_2=\emptyset$. 
As a direct inference of \cref{sec3:thm1}, we know that there exists a semialgebraic function $h(\xx)$ of the form in \cref{eq:h} such that $h(\xx)>0$ is an interpolant of $\phi(\xx)$ and $\psi(\xx)$.

\subsection{Interpolant between $\phi(\xx,\yy)$ and $\psi(\xx,\zz)$}
\label{ssec:3.2}

Let $\phi(\xx,\yy)$ and $\psi(\xx,\zz)$ be given in \cref{problem:1}.
We denote by $T_\phi\subseteq \RR^{r_1+r_2}$ and $T_\psi\subseteq \RR^{r_1+r_3}$ the semialgebraic sets corresponding to $\phi$ and $\psi$, i.e., 
\vspace{-2mm}
\begin{align}
    T_{\phi} &\coloneqq \bigcup_{k=1}^{K_\phi} S_k, \text{ with }S_k\coloneqq \{(\xx, \yy)\in\RR^{r_1+r_2} \mid \bigwedge_{i=1}^{m_k} f_{k,i}(\xx,\yy) \ge 0 \}, \label{eq:phi-T}\\
    T_{\psi} &\coloneqq \bigcup_{k'=1}^{K_\psi} S'_{k'}, \text{ with }S'_{k'}\coloneqq \{(\xx, \zz) \in\RR^{r_1+r_3}\mid \bigwedge_{j=1}^{n_{k'}} g_{k',j}(\xx,\zz) \ge 0 \}. \label{eq:psi-T}
\end{align}
\vspace{-2mm}

Since an interpolant contains only common symbols of $\phi$ and $\psi$, \cref{problem:1} can be reduced to finding a function $h(\xx)$ such that $h(\xx)=0$ separates the two projection sets $P_{\xx}(T_\phi)\coloneqq \{\xx \in \RR^{r_1}\mid \exists \yy.~(\xx,\yy)\in T_{\phi}\}$ and 
$P_{\xx}(T_\psi)\coloneqq \{\xx \in \RR^{r_1} \mid \exists \zz.~(\xx,\zz)\in T_{\psi}\}$.
We have the following theorem as a direct consequence  of \cref{sec3:thm1}.

\begin{theorem} \label{thm:main}
Let $\phi(\xx,\yy)$ and $\psi(\xx,\zz)$ be defined in \cref{problem:1}, and let $P_{\xx}(T_{\phi})$ and $P_{\xx}(T_\psi)$ be defined above.
Let $T_1=\cl(P_{\xx}(T_{\phi}))$ and $T_2=\cl(P_{\xx}(T_\psi))$.
Assume $\tilde{T}_1\cap \tilde{T}_2=\emptyset$.
Then there exists a semialgebraic function $h(\xx)$ of the form in \cref{eq:h} such that
\begin{equation}
    \forall \xx\in P_{\xx}(T_{\phi}).~h(\xx)>0 \text{ and }\forall \xx\in P_{\xx}(T_{\psi}).~-h(\xx)>0.
\end{equation}
As a consequence, $h(\xx)>0$ is a semialgebraic interpolant of $\phi(\xx,\yy)$ and $\psi(\xx,\zz)$.
\end{theorem}


\begin{remark}
Note that in \cref{thm:main}, we need to consider the closures $\cl(P_{\xx}(T_{\phi}))$ and $\cl(P_{\xx}(T_\psi))$ rather than
$P_{\xx}(T_{\phi})$ and $P_{\xx}(T_\psi)$ themselves. 
The reason lies in that the projections of closed semialgebraic sets are not necessarily closed.
For example, consider $\phi(\xx,\yy)\coloneqq x_1x_2-1\ge 0\wedge x_2\ge 0$ with $\xx=x_1$ and $\yy=x_2$. Then $P_{\xx}(T_{\phi}) = \{x_1\mid x_1>0\}$ is an open set.
\end{remark}

%% file: sections/4-sos.tex
\vspace{-3mm}
\section{Sum-of-Squares Formulation}
\label{sec:sos}
\vspace{-2mm}
In this section, we provide SOS characterizations for polynomial and semialgebraic interpolants.
For simplicity, we will focus on the case where $\phi$ and $\psi$ are conjunctions of polynomial inequalities given by 
\begin{equation} \label{eq:phi-psi-basic}
    \phi(\xx,\yy) \coloneqq \bigwedge_{i=1}^m f_i(\xx,\yy) \ge 0\text{ and }
    \psi(\xx,\zz) \coloneqq \bigwedge_{j=1}^n g_j(\xx,\zz) \ge 0,
\end{equation}
where $\xx\in \RR^{r_1}$, $\yy\in \RR^{r_2}$, and $\zz\in \RR^{r_3}$.
Extending to the general case is straightforward.

\subsection{SOS Characterization for Polynomial Interpolants}\label{sec:4.1}
In this part, we provide an SOS characterization for polynomial interpolants based on homogenization.
We prove that the characterization is sound and weakly complete. 
Furthermore, we provide a concrete example to show that our new characterization is strictly more expressive than the one in \cite{gan2020nonlinear}.

\begin{theorem}[Weak Completeness]\label{thm:sos-formulation} 
Let $\phi$, $\psi$ be defined as in \cref{eq:phi-psi-basic} and let $S_\phi$, and $S_\psi$ be the basic semialgebraic sets corresponding to $\phi$ and $\psi$.
Let $\tilde{f}_{m+1}= x_0$, $\tilde{g}_{n+1}= x_0$, $\tilde{f}_{m+2}= x_0^2+\|\xx\|^2+\|\yy\|^2-1$, and $\tilde{g}_{n+2}= x_0^2+\|\xx\|^2+\|\zz\|^2-1$.
If $h(\xx)\in \RR[\xx]$ is a polynomial interpolant function of $\phi$ and $\psi$, 
then the homogenized polynomial $\tilde{h}(x_0, \xx)$ satisfies, for arbitrarily small $\epsilon>0$,
\vspace{-2mm}
\begin{equation}\label{eq:sos-formulation}
\begin{aligned}
    \tilde{h}(x_0,\xx) + \epsilon
    &~=~\sigma_0+\sum_{i=1}^{m+2}\sigma_i\tilde{f}_i(x_0,\xx,\yy), \\
    -\tilde{h}(x_0,\xx) + \epsilon &~=~\tau_0+\sum_{j=1}^{n+2}\tau_j\tilde{g}_j(x_0,\xx,\zz),    
\end{aligned}
\end{equation}
for some $\sigma_i\in\mathrm{\Sigma}[x_0,\xx,\yy]$, $i=0,\ldots,m+1$, $\sigma_{m+2}\in\RR[x_0,\xx,\yy]$, $\tau_i\in\mathrm{\Sigma}[x_0,\xx,\zz]$, $i=0,\ldots,n+1$, $\tau_{n+2}\in\RR[x_0,\xx,\zz]$.
\end{theorem}


\begin{remark}\label{remark:convergence}
    In \cref{eq:sos-formulation}, we add a small quantity $\epsilon>0$ to the left-hand sides in order to invoke \cref{thm:putinar}. 
    The ideal case is $\epsilon=0$. 
    Fortunately, in most practice circumstances, we can safely set $\epsilon=0$ when the \emph{finite convergence property} \cite[Thm.~1.1]{nie2014optimality} holds. Indeed, the finite convergence property is generically true and is violated only when $h$ and $S_{\phi}$ (or $S_{\psi}$) are of certain singular forms.

\end{remark}

\begin{theorem}[Soundness]\label{thm:sound}
    Let $\phi$ and $\psi$ be defined as in \cref{eq:phi-psi-basic}. Suppose that $h(\xx)$ is a polynomial such that its homogenization $\tilde{h}(x_0,\xx)$ satisfies \cref{eq:sos-formulation} with $\epsilon=0$.
    Assume $\phi\wedge h(\xx)=0\models \bot$ and $\psi \wedge h(\xx)=0\models \bot$. Then $h(\xx)$ is an interpolant function of $\phi$ and $\psi$.
\end{theorem}


Now we compare our characterization \cref{eq:sos-formulation} with \cite[Thm.~5]{gan2020nonlinear} which states that if $S_{\phi}$ and $S_\psi$ are of the Archimedean form, then a polynomial interpolant function $h(\xx)$ can be expressed as  
\begin{equation}\label{eq:sos-formulation-20}
    \begin{aligned}
        h(\xx) - 1  &= \sigma_0 + \sum_{i=0}^{m} \sigma_i f_i(\xx,\yy),\\ 
        - h(\xx) - 1  &= \tau_0 + \sum_{j=0}^{n} \tau_j g_j(\xx,\zz),\\ 
    \end{aligned}
\end{equation}
for some $\sigma_i \in \mathrm{\Sigma}[\xx,\yy]$, $i=0,\dots,m$ and $\tau_j\in \mathrm{\Sigma}[\xx,\zz]$, $j=0,\dots,n$.
Clearly, since \cref{thm:sos-formulation} removes the restriction of the Archimedean condition, our characterization is \emph{strictly} more expressive.

Let $M(x_1,x_2)\coloneqq x_1^4x_2^2+x_1^2x_2^4-3x_1^2x_2^2+1$ be the Motzkin polynomial.
It is well known that $M(x_1,x_2)$ is nonnegative but is not an SOS.

\begin{proposition}\label{prop:motzkin}
    The polynomial $M(x_1,x_2)+1$ is positive but is not an SOS.            
\end{proposition}


\begin{example}
    Let $\phi\coloneqq 1\ge 0 (= \top)$ and $\psi\coloneqq -1 \ge 0 (= \bot)$.
    By \cref{prop:motzkin}, the polynomial $M(x_1,x_2)+1$ is an interpolant function of $\phi$ and $\psi$ but does not admit a representation as in \cref{eq:sos-formulation-20}, i.e., the program
    \begin{align*}
        \text{find} \quad &\sigma_0 \in \mathrm{\Sigma}[x_1,x_2]\\
        \rm{s.t.} \quad & x_1^4x_2^2+x_1^2x_2^4-3x_1^2x_2^2+2 =\sigma_0
    \end{align*}
    is not feasible.
    However, a numerical solution to the following program:
    \begin{align*}
        \text{find} \quad &\sigma_0, \sigma_1 \in \mathrm{\Sigma} [x_0,x_1,x_2], \sigma_2\in \RR[x_0,x_1,x_2]\\
        \rm{s.t.} \quad & x_1^4x_2^2 + x_1^2x_2^4 - 3x_1^2x_2^2x_0^2 + 2x_0^6= \sigma_0 + \sigma_1 x_0 + \sigma_2 (1-x_0^2-x_1^2-x_2^2).
    \end{align*}
    can be obtained by employing the Julia package \textsc{TSSOS} \cite{tssos_tool} and the SDP solver \textsc{Mosek} \cite{mosek}.
    Therefore, the polynomial $M(x_1,x_2)+1$ admits a representation as in \cref{eq:sos-formulation} with $\epsilon=0$.
\end{example}

\subsection{SOS Characterization for Semialgebraic Interpolants}
\label{sec:4.2}
Let $h(\xx)$ be a semialgebraic interpolant function of the form in \cref{eq:h} and let $w$ be a fresh variable.
Though $h(\xx)$ is not a polynomial, it can be equivalently represented by a polynomial $l(\xx, w)=h_1(\xx)+w\cdot h_2(\xx)\in \RR[\xx,w]$ with additional polynomial constraints $w^2=1+\|\xx\|^2$ and $w\ge 0$.
Adopting this idea, we have the following completeness theorem.
The soundness result for the semialgebraic case is omitted, as it is essentially the same as \cref{thm:sound}.

\begin{theorem}[Completeness]\label{thm:completeness} 
Let $\phi$, $\psi$ be defined as in \cref{eq:phi-psi-basic} and let $S_\phi$, and $S_\psi$ be the basic semialgebraic sets corresponding to $\phi$ and $\psi$.
Let $S_1 = \cl(P_{\xx}(S_\phi))$, $S_2 = \cl(P_{\xx}(S_\psi))$,
$\tilde{f}_{m+1}=\tilde{g}_{n+1}= x_0$, $\tilde{f}_{m+2}=\tilde{g}_{m+2}=w$, $\tilde{f}_{m+3}= x_0^2+\|\xx\|^2+w^2+\|\yy\|^2-1$, and $\tilde{g}_{n+3}= x_0^2+\|\xx\|^2+w^2+\|\zz\|^2-1$, $\tilde{f}_{m+4}=\tilde{g}_{n+4}= x_0^2+\|\xx\|^2-w^2$.
Assume that the following two conditions hold: 
(1) $\tilde{S}_1\cap \tilde{S_2} = \emptyset$;
(2) $S_\phi$ and $S_\psi$ are closed at $\infty$.
Then there exists a semialgebraic interpolant function $h(\xx)$ of the form in \cref{eq:h} such that the polynomial $l(\xx, w)=h_1(\xx)+w\cdot h_2(\xx)\in \RR[\xx,w]$ satisfies, for arbitrarily small $\epsilon>0$,
\vspace{-2mm}
\begin{equation}\label{eq:complete}
\begin{aligned}
    \tilde{l}(x_0, \xx, w) + \epsilon
    &~=~\sigma_0+\sum_{i=1}^{m+4}\sigma_i\tilde{f}_i(x_0,\xx,\yy), \\
    -\tilde{l}(x_0, \xx, w)  + \epsilon &~=~\tau_0+\sum_{j=1}^{n+4}\tau_j\tilde{g}_j(x_0,\xx,\zz),    
\end{aligned}
\end{equation}
for some $\sigma_i\in\mathrm{\Sigma}[x_0,\xx,\yy,w]$, $i=0,\ldots,m+2$, $\sigma_{m+3},\sigma_{m+4}\in\RR[x_0,\xx,\yy,w]$, $\tau_i\in\mathrm{\Sigma}[x_0,\xx,\zz,w]$, $i=0,\ldots,n+2$, $\tau_{n+3},\tau_{n+4}\in\RR[x_0,\xx,\zz,w]$.
\end{theorem}

We want to emphasize that \cref{thm:completeness} is a stronger result than \cref{thm:sos-formulation}, in the sense that \cref{thm:completeness} guarantees the existence of a semialgebraic interpolant (as per \cref{thm:main}), which is not the case for polynomial interpolants in \cref{thm:sos-formulation}.



%% file: sections/5-interpolation.tex
\section{Synthesizing Interpolant via SOS Programming}
\label{sec:5}
\vspace{-2mm}
In this section, we propose an SOS programming procedure to synthesize polynomial and semialgebraic interpolants.
Concrete examples are provided to demonstrate the effectiveness and efficiency of our method.
For all examples, existing approaches \cite{DXZ13,GDX16,Gao14} are not applicable due to their restrictions on formulas, and the method in \cite{gan2020nonlinear} also fails to produce interpolants of specified degrees.
All experiments were conducted on a Mac lap-top with Apple M2 chip and 8GB memory.
We use the Julia package \textsc{TSSOS} \cite{tssos_tool} to formulate SOS programs and rely on the SDP solver \textsc{Mosek} \cite{mosek} to solve them.
All numerical results are symbolically verified using \textsc{Mathematica} to be real interpolants and the portraits can be found in \cref{app:fig}. 
The scripts are publicly available
\footnote{\url{https://github.com/EcstasyH/Interpolation}.}.


\textbf{Synthesizing Polynomial Interpolants}: Let $T_\phi$, $T_\psi$, $S_k$, and $S_{k'}$ be defined as in \cref{eq:phi-T} and \cref{eq:psi-T}. 
By treating $S_k$ and $S_k'$ respectively as $S_\phi$ and $S_\psi$ in \cref{thm:sos-formulation},
the problem of synthesizing a polynomial interpolant for $\phi$ and $\psi$ is reduced to solving the following SOS program:
\begin{equation} \label{eq:sos-program-general}
  \left\{
    \begin{aligned}\allowdisplaybreaks
    \text{find} \quad & h(\xx)\\
    \rm{s.t.} \quad 
    &\tilde{h}(x_0,\xx)=\sigma_{k,0}+\sum_{i=1}^{m_k+2} \sigma_{k,i} \tilde{f}_{k,i} \quad \text{for }k=1,\dots,K_{\phi},\\
    &-\tilde{h}(x_0,\xx)=\tau_{k',0}+\sum_{j=1}^{n_{k'}+2} \tau_{k',j} \tilde{g}_{k',j} \quad \text{for }k'=1,\dots,K_{\psi},\\
    &\sigma_{k,0},...,\sigma_{k,m+1}\in \mathrm{\Sigma}[x_0,\xx,\yy], \sigma_{k,m+2}\in \RR[x_0,\xx,\yy], \\
    & \qquad  \text{for }k=1,\dots,K_{\phi},\\
    &\tau_{k',0},...,\tau_{k',n+1}\in \mathrm{\Sigma}[x_0,\xx,\zz], 
    \tau_{k',n+2}\in \RR[x_0,\xx,\yy],\\
    & \qquad \text{for }k=1,\dots,K_{\psi},
  \end{aligned}
  \right.
\end{equation}
where $\tilde{f}_{k,m+1}=\tilde{g}_{k',n+1}=x_0$, 
$\tilde{f}_{k,m+2} =x_0^2+\|\xx\|^2+\|\yy\|^2-1$, $\tilde{g}_{k',n+2}=x_0^2+\|\xx\|^2+\|\zz\|^2-1$ 
for $k=1,\dots,K_{\phi}$ and $k'=1,\dots,K_{\psi}$.

As \cref{thm:sound} suggests, a solution $h(\xx)$ to the above program only ensures that $\phi\models h(\xx)\ge 0$ and $\psi \models -h(\xx)\le 0$.
Nevertheless, since numerical solvers are unable to distinguish $\ge$ from $>$, the equalities are usually not attainable for a numerical solution\footnote{For example, SDP solvers based on interior-point methods typically return strictly feasible solutions.}.
Therefore, we can view the SOS program \cref{eq:sos-program-general} as a sound approach for computing $h(\xx)$, while completeness follows from verifying the conditions discussed in \cref{remark:convergence}.
 
In practice, we solve the program \cref{eq:sos-program-general} by solving a sequence of SDP relaxations which are obtained by restricting the highest degree of involved polynomials.
Concretely speaking, suppose that we would like to find a polynomial interpolant function $h(\xx)$ of degree~$d$, we set the template of $h(\xx)$ to be $h(\xx)= \sum_{|\ba|\le d} c_{\ba} \xx^{\ba}$,
where $\ba=(\alpha_1, ..., \alpha_{r_1})\in\NN^{r_1}$, $|\ba|=\alpha_1+ \cdots \alpha_{r_1}$, and $c_{\ba}\in \RR$ are coefficients to be determined.
Then, the homogenization of $h(\xx)$ is $\tilde{h}(x_0,\xx)= \sum_{|\ba|\le d} c_{\ba} x_0^{d-|\ba|} \xx^{\ba}$.

Given a relaxation order $s\in \NN$ with $2s\ge d$, we set the degrees of the remaining unknown polynomials $\sigma_i,\tau_j$ appropriately to ensure that the maximum degree of polynomials involved in \cref{eq:sos-program-general} equals $2s$.
We refer to the resulting program as the $s$-th relaxation of \cref{eq:sos-program-general}, which can be translated into an SDP and can be numerically solved in polynomial time.
If the $s$-th relaxation is solvable, it yields a solution $h(\xx)$ that serves as a polynomial interpolant function of $\phi$ and $\psi$.
If it is not solvable, we then increase the relaxation order $s$ to obtain a tighter relaxation, or alternatively, we can increase the degree $d$ of $h(\xx)$ to search for interpolants of higher degree.

\begin{example}[adapted from \cite{chen19cade}]\label{ex:2}
    Let $\xx=(x,y)$ and $\yy=\zz=\emptyset$, i.e., there is no uncommon variables.
    We define the following polynomials:
    \begin{align*}
        &f_1 = 11 - x^4 + 0.1y^4,\quad  &&f_2 = y^3,\\
        &f_3 = 0.9025 - (x-1)^4 - y^4,\quad &&f_4 = (x-1)^4 + y^4 - 0.09,\\
        &f_5 = (x+1)^4 + y^4 - 1.1025,\quad &&f_6 = 0.04 - (x+1)^4 - y^4,\\
        &g_1 = 11 - x^4 + 0.1y^4,\quad &&g_2 = -y^3,\\
        &g_3 = 0.9025 - (x+1)^4 - y^4, \quad &&g_4 = (x+1)^4 + y^4 - 0.09,\\
        &g_5 = (x-1)^4 + y^4 - 1.1025, \quad &&g_6 = 0.04 - (x-1)^4 - y^4.
    \end{align*}
    Let $\phi$ and $\psi$ be defined by
    \begin{align*}
        &\phi \coloneqq (f_1\ge0 \wedge f_2\ge 0 \wedge f_4\ge 0\wedge f_5\ge 0) \vee (f_3\ge0 \wedge f_4\ge 0 \wedge f_5\ge 0) \vee (f_6\ge 0),\\ 
        &\psi \coloneqq (g_1\ge 0 \wedge g_2\ge 0 \wedge g_4 \ge 0 \wedge g_5\ge 0) \vee (g_3\ge 0 \wedge g_4\ge 0 \wedge g_5 \ge 0) \vee (g_6\ge 0).
    \end{align*}
    Set the degree of the polynomial interpolation function $h(x,y)$ to $7$.
    It takes 0.16 seconds to solve the $4$-th relaxation \cref{eq:sos-program-general}, yielding the solution 
    \begin{equation*}
        h(x,y) =  - 0.00153942 y + 0.03053692x  + \cdots  + 0.06109453 x^6 y + 0.01643640 x^7,
    \end{equation*}
    where the coefficients have been scaled so that the largest absolute value is $1$.
\end{example}


\textbf{Synthesizing Semialgebraic Interpolants}:
Similarly, the synthesis of a semialgebraic interpolant is reduced to solving the following SOS program:
\begin{equation} \label{eq:sos-program-nonpoly}
  \left\{\begin{aligned}
    \text{find} \quad & h_1(\xx), h_2(\xx)\\
    \rm{s.t.} \quad 
    & l(\xx, w) = h_1(\xx)+w\cdot h_2(\xx),\\
    &\tilde{l}(x_0,\xx, w)=\sigma_{k,0}+\sum_{i=1}^{m_k+4} \sigma_{k,i} \tilde{f}_{k,i} \quad \text{for }k=1,\dots,K_{\phi},\\
    &-\tilde{l}(x_0,\xx, w)=\tau_{k',0}+\sum_{j=1}^{n_{k'}+4} \tau_{k',j} \tilde{g}_{k',j} \quad \text{for }k'=1,\dots,K_{\psi},\\
    &\sigma_{k,0},...,\sigma_{k,m+2}\in \mathrm{\Sigma}[x_0,\xx,\yy], \sigma_{k,m+3}, \sigma_{k,m+4}\in \RR[x_0,\xx,\yy], \\
    & \qquad  \text{for }k=1,\dots,K_{\phi},\\
    &\tau_{k',0},...,\tau_{k',n+2}\in \mathrm{\Sigma}[x_0,\xx,\zz], 
    \tau_{k',n+3}, \tau_{k',n+4}\in \RR[x_0,\xx,\yy],\\
    & \qquad \text{for }k=1,\dots,K_{\psi},
  \end{aligned}\right.
\end{equation}
where $\tilde{f}_{k,m+1}=\tilde{g}_{k',n+1}=x_0$,
$\tilde{f}_{k,m+2}=\tilde{g}_{k',n+2}=w$,
$\tilde{f}_{k,m+3} = x_0^2+\|\xx\|^2+w^2+\|\yy\|^2-1$, $\tilde{g}_{k',n+3}= x_0^2+\|\xx\|^2+w^2+\|\zz\|^2-1$,
$\tilde{f}_{k,m+4} = \tilde{g}_{k',n+4}=x_0^2+\|\xx\|^2-w^2$,
for $k=1,\dots,K_{\phi}$ and $k'=1,\dots,K_{\psi}$.

By \cref{thm:completeness}, if a feasible solution $(h_1, h_2)$ of \cref{eq:sos-program-nonpoly} is found, then $h(\xx)=h_1(\xx)+\sqrt{\|\xx\|^2+1}\cdot h_2(\xx)$ is a semialgebraic interpolant function for $\phi$ and $\psi$.
In practice, w.l.o.g., we can assume that $h_1$ and $h_2$ are of the same degree~$d$ and solve SDP relaxations of \cref{eq:sos-program-nonpoly}.
The soundness result is similar to that of \cref{eq:sos-program-general}, requiring that $h(\xx)=0$ is not attainable over $T_{\phi}$ and $T_{\psi}$.

\begin{example}\label{ex:4}
Let $\xx=(x,y)$, $\yy=\zz=\emptyset$. We define
\begin{align*}
    \phi(x,y) &= 8xy-(x^2-y^3)^2\ge 0 \wedge x^2+y^2-1\ge 0,\\
    \psi(x,y) &= -12.5xy-(x^2+y^2)^2\ge 0 \wedge x^2+y^2-1\ge 0.
\end{align*}
Let the degree of $h_1(\xx)$ and $h_2(\xx)$ to be $3$,  
a solution to the $2$-th relaxation of \cref{eq:sos-program-nonpoly} is found in 0.02 seconds: 
\begin{align*}
    h_1 &= - 0.04402209 - 0.00093184 y + 0.01446436 x + \dots - 0.03703461 x^3,\\
    h_2 &= 0.05644318 - 0.01305178 y + 0.02407258 x +\cdots + 0.23199837 x^2.
\end{align*}
As a comparison, solving \cref{eq:sos-program-general} fails to produce a polynomial interpolant function of degree $3$, but succeeds at degree $4$.
\end{example}

\vspace{-1mm}
\begin{example}\label{ex:1}
Let $\xx=(x,y,z)$, $\yy=\emptyset$, and $\zz=(r,R)$. We define
\vspace{-1mm}
\begin{align*}
    \phi(x,y,z) & \coloneqq~ 1 + 0.1z^4 - x^4 - y^4 \ge 0 \wedge 10z^4 - x^4 - y^4 \ge 0,\\
    \psi(x,y,z,r,R) & \coloneqq~ 4R^2(x^2+y^2) - (x^2+y^2+z^2+R^2-r^2)^2 \ge 0 \\
    & \qquad \wedge 6\ge R\ge 4 \wedge 1\ge r\ge 0.5,
\end{align*}
where $\exists r, \exists R.~\psi(x,y,z,r,R)$ describes the set of interior points of a 3-dimensional torus with unknown minor radius~$r\in [0.5,1]$ and major radius~$R\in [4,6]$. 
By solving \cref{eq:sos-program-general} and \cref{eq:sos-program-nonpoly}, we obtain a polynomial interpolant
\begin{equation*}
    h_p(x,y,z) = 1.0 - 0.35507338 x^2 - 0.35507338 y^2 + 0.45264895 z^2,
\end{equation*}
and a semialgebraic interpolant function with
\begin{align*}
    h_1(x,y,z) &= 0.98004189 - 0.26291972 x^2 - 0.26291978 y^2 + 0.417581644 z^2,\\
    h_2(x,y,z) &= 1.0 - 0.51670759 x^2 - 0.51670759 y^2 + 0.60569150 z^2.
\end{align*}
As a comparison, \cite{gan2020nonlinear} fails to produce an interpolant of degree less than $4$.    
\end{example}

\vspace{-2mm}

%% file: sections/6-conclusion.tex
\section{ Conclusions and Future Work}\label{sec:con}

In this paper, we have addressed the problem of synthesizing Craig interpolants for two general polynomial formulas.
By combining the polynomial homogenization techniques with the approach from \cite{gan2020nonlinear}, we have presented a complete SOS characterization of semialgebraic (and polynomial) interpolants.
Compared with existing works, our approach removes the restrictions on the form of formulas and is applicable to any polynomial formulas, especially when variables have unbound domains. Moreover, sparsity of polynomial formulas can be exploited to improve the scalability of our approach \cite{huang2024sparse,magron2023sparse}.

Our Craig interpolation synthesis technique offers broad applicability in various verification tasks.
It can be used as a sub-procedure, for example, in CEGAR-based model checking for identifying counterexamples \cite{mcmillan05}, in bounded model checking for generating proofs \cite{komuravelli16fmsd},
in program verification for squeezing invariants \cite{LSXLSH2017}, and in SMT for reasoning about nonlinear arithmetic \cite{jovanovic21cav}. 
Compared with existing algorithms, our SDP-based algorithm is efficient and provides a relative completeness guarantee. 
However, the practical implementation is not a trivial undertaking, as it requires suitable strategies for storing numerical interpolants and taming numerical errors \cite{roux18fmsd}.
This remains an ongoing work of our research.

%% file: sections/appendix.tex
\newpage
\appendix

\section{Omitted Proofs}
\label{app}

\subsubsection{Proof of \cref{sec3:lm1}}
\begin{proof}
Suppose $S_1=\{\xx\in\RR^{r} \mid p_1(\xx) \ge 0,\ldots,p_m (\xx) \ge 0\}$, $S_2=\{\xx\in\RR^{r} \mid q_1(\xx) \ge 0,\ldots,q_n (\xx) \ge 0\}$. Then, the intersection $S_1\cap S_2$ is the set
\begin{equation*}
  S=\{\xx\in\RR^{r} \mid p_1(\xx) \ge 0,\ldots,p_m (\xx) \ge 0, q_1(\xx) \ge 0,\ldots,q_n (\xx) \ge 0\}.
\end{equation*}
By the definitions in \cref{eq:tS} and \cref{eq:Sinf}, it is easy to see $S^{(\infty)}=S^{(\infty)}_1\cap S^{(\infty)}_2$ and $\tilde{S}=\tilde{S}_1\cap\tilde{S}_2$. 
From the conditions in this lemma, we have
$S=\emptyset$ and $S^{(\infty)}=\emptyset$.
Therefore, from \cref{prop:3} we have $\tilde{S}=\tilde{S}^h$ and from \cref{property:hom} we have $\tilde{S}^h=\emptyset$, which gives $\tilde{S}_1\cap\tilde{S}_2=\tilde{S}=\emptyset$.
\end{proof}
\vspace{-5mm}
\subsubsection{Proof of \cref{sec3:basic-semi}}
\begin{proof}
For each $i=1,\ldots,a$, we invoke the proof of \cref{sec3:prop2} by treating $S_i$ and $T_2$ respectively as $S_1$ and $S_2$, obtaining a polynomial $g_i\in\RR[x_0,\xx]$ such that
\begin{equation*}
    \forall (x_0,\xx)\in \tilde{S}_i.~ g_i(x_0,\xx)>0 \text{ and }\forall (x_0,\xx)\in \tilde{T}_2.~ -g_i(x_0,\xx)>0.
\end{equation*}
Since $\tilde{T}_2$ is compact, using arguments similar to those in \cite[Lem.~3]{gan2020nonlinear},
we can construct a new polynomial $g\in\RR[x_0, \xx]$ satisfying \cref{sec3:eq3} from $g_1,\ldots,g_a$.
\end{proof}
\vspace{-5mm}
\subsubsection{Proof of \cref{sec3:thm1}}
\begin{proof}
By \cref{sec3:basic-semi}, there exists a polynomial $g_j(x_0,\xx)$ such that
\begin{equation*}
    \forall (x_0,\xx)\in \tilde{T}_1.~ g_j(x_0,\xx)>0 \text{ and }\forall (x_0,\xx)\in \tilde{S}'_j.~ -g_j(x_0,\xx)>0,
\end{equation*}
for $j=1,\ldots,b$.
Since $\tilde{T}_1$ is compact, using arguments similar to those for \cite[Lem.~3]{gan2020nonlinear}, we can show that there exists a polynomial $g\in\RR[x_0, \xx]$ such that
\begin{equation*}
    \forall (x_0,\xx)\in \tilde{T}_1.~ g(x_0,\xx)>0 \text{ and }\forall (x_0,\xx)\in \tilde{T}_2.~ -g(x_0,\xx)>0.
\end{equation*}
Using similar arguments as in the proof of \cref{sec3:prop2}, we find that
the semialgebraic function $h(\xx)\coloneqq(\sqrt{\|\xx\|^2+1})^{\deg(g)}g(\frac{1}{\sqrt{\|\xx\|^2+1}},\frac{\xx}{\sqrt{\|\xx\|^2+1}})$ is of the desired form and satisfies \eqref{sec3:eq4}.
\end{proof}
\vspace{-5mm}
\subsubsection{Proof of \cref{thm:main}}
\begin{proof}
According to the Tarski-Seidenberg theorem \cite{bochnak13book}, the projection of a semialgebraic set is also a semialgebraic set. This indicates that $P_{\xx}(T_{\phi})$ and $P_{\xx}(T_{\psi})$ are both semialgebraic sets.
By treating $\cl(P_{\xx}(T_{\phi}))$ and $\cl(P_{\xx}(T_\psi))$ respectively as $T_1$ and $T_2$ and invoking \cref{sec3:thm1}, we obtain the desired result.
\end{proof}
\vspace{-5mm}
\subsubsection{Proof of \cref{thm:sos-formulation}}
\begin{proof}
Since $h(\xx)$ is an interpolant function, we have $h(\xx)>0$ over $S_{\phi}$ and $h(\xx)<0$ over $S_{\psi}$.
By \cref{sec2:lm}, we know $\tilde{h}(\xx)\ge0$ on $\tilde{S}_\phi$ and $-\tilde{h}(\xx)\ge0$ on $\tilde{S}_\psi$.
The desired conclusion then follows from the fact that both $\tilde{S}_\phi$ and $\tilde{S}_\psi$ are of the Archimedean form and \cref{thm:putinar}.
\end{proof}
\vspace{-5mm}
\subsubsection{Proof of \cref{thm:sound}}
\begin{proof}
    Since $\tilde{h}(x_0,\xx)$ satisfies \cref{eq:sos-formulation}, we have $\tilde{h}(x_0,\xx)\ge 0$ on $\tilde{S}_{\phi}$ and $\tilde{h}(x_0,\xx)\le 0$ on $\tilde{S}_\psi$.
    By \cref{property:hom}, we have $h(\xx)\ge 0$ on $S_{\phi}$ and $h(\xx)\le 0$ on $S_{\psi}$, where the assumption ensures that $=$ is not attainable.  
    Therefore, the theorem is proved.
\end{proof}
\vspace{-5mm}
\subsubsection{Proof of \cref{prop:motzkin}}
\begin{proof}
    Suppose on the contrary that $M(x_1,x_2)+1=\sum_{i} p_i^2$.
    Since $M(x_1,x_2)+1$ has degree $4$ in $x_1$ and in $x_2$, monomials with the power of $x_1$ or $x_2$ greater than or equal to $3$ cannot occur in $p_i$.
    Similarly, since $x_1^4x_2^4$, $x_1^2$, $x_2^2$ do not occur in $M(x_1,x_2)+1$, the monomials $x_1^2x_2^2$, $x_1$, $x_2$ cannot occur in $p_i$. 
    Therefore, each $p_i$ contains only a subset of the monomials $x_1^2x_2$, $x_1x_2^2$, $xy$ and $1$.
    However, this implies that the coefficient of $x_1^2x_2^2$ in $M(x_1,x_2)+1$ must be nonnegative, which is a contradiction.
\end{proof}
\vspace{-5mm}
\subsubsection{Proof of \cref{thm:completeness}}
\begin{proof}
    According to \cref{thm:main}, the existence of $h(\xx)$ is ensured by the first condition.
    So we have $h(\xx)>0$ over $S_{\phi}$, which implies that $l(\xx, w)>0$ over the semialgebraic set
    $S_w\coloneqq \{(\xx,\yy, w)\in \RR^{r_1+r_2+1}\mid \phi(\xx,\yy), w^2=\|\xx\|^2+1, w\ge 0\}$. 
    By \cref{sec2:lm}, we obtain $\tilde{l}(x_0,\xx, w)\ge 0$ over $\tilde{S}_w$.
    Since $\tilde{S}_w$ is of the Archimedean form, the first equation of \cref{eq:complete} is obtained by applying \cref{thm:putinar}.
    The proof of the second equation is similar.    
\end{proof}

\section{Portraits}
\label{app:fig}

\begin{figure}[h]
    \captionsetup{font={small}}
    \centering
    \begin{subfigure}[b]{0.30\textwidth}
        \centering
        \includegraphics[width=\textwidth]{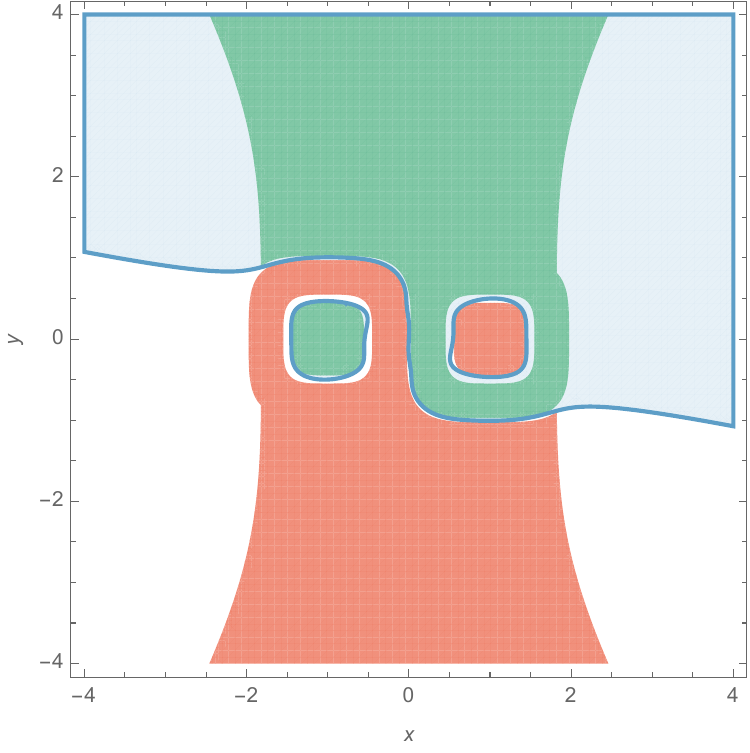}
        \caption{\cref{ex:2}}
    \end{subfigure}
    \hfill
    \begin{subfigure}[b]{0.30\textwidth}
        \centering
        \includegraphics[width=\textwidth]{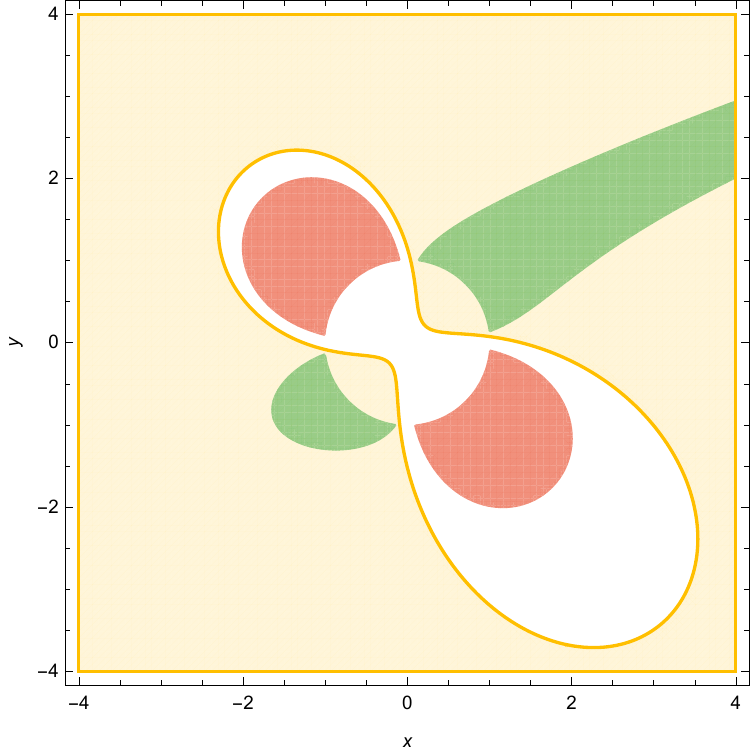}
        \caption{\cref{ex:4}}
    \end{subfigure}
    \hfill
    \begin{subfigure}[b]{0.30\textwidth}
        \centering
        \includegraphics[width=\textwidth]{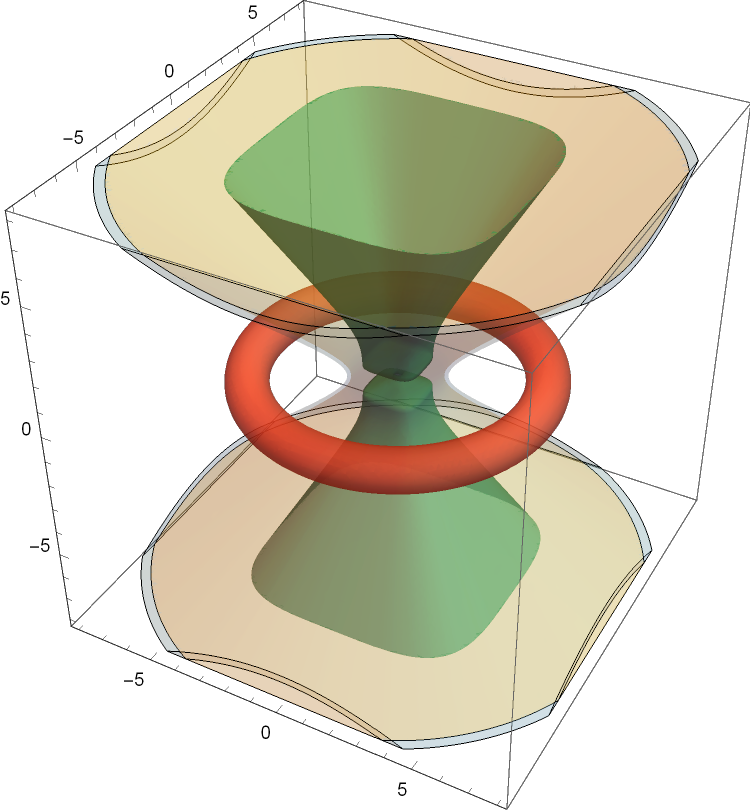}
        \caption{\cref{ex:1}}
    \end{subfigure}
    \\
    \scriptsize{
    green region: projection of $\phi$;
    red region: projection of $\psi$;
    light blue/yellow region: polynomial/semialgebraic interpolant. For \cref{ex:1}, $\psi$ is plotted for $r=0.75$ and $R=5$.
    }
    \caption{Portraits of Examples.}
\end{figure}